\documentclass[10pt,reqno]{amsart}
     \makeatletter
     \def\section{\@startsection{section}{1}%
     \z@{.7\linespacing\@plus\linespacing}{.5\linespacing}%
     {\bfseries%\normalfont\scshape
     \centering
     }}
     \def\@secnumfont{\bfseries}
     \makeatother
\newtheorem{theorem}{Theorem}[section]

\theoremstyle{definition}

\theoremstyle{remark}

\numberwithin{equation}{section}
\usepackage{amsmath}
\usepackage{amssymb}
\usepackage{amsfonts}
\usepackage{graphicx}
\usepackage{subfigure}
\begin{document}

\title[Real gas flows issued from a source]{Real gas flows issued from a source}

\author{Valentin Lychagin}
\address{Valentin Lychagin: V.A. Trapeznikov Institute of Control Sciences, Russian Academy of Sciences, 65 Profsoyuznaya Str., 117997 Moscow, Russia, Department of Mathematics, The Arctic University of Norway, Postboks 6050, Langnes 9037, Tromso, Norway}
\email{valentin.lychagin@uit.no}

\author{Mikhail Roop}
\address{Mikhail Roop: V.A. Trapeznikov Institute of Control Sciences, Russian Academy of Sciences, 65 Profsoyuznaya Str., 117997 Moscow, Russia, Department of Physics, Lomonosov Moscow State University, Leninskie Gory, 119991 Moscow, Russia}
\email{mihail\underline{ }roop@mail.ru}



\begin{abstract}
Stationary adiabatic flows of real gases issued from a source of given intensity are studied. Thermodynamic
states of gases are described by Legendrian or Lagrangian manifolds. Solutions of Euler equations are given implicitly for any equation of state and the behavior of solutions of the Navier-Stokes equations with the viscosity considered as a small parameter is discussed. For different intensities of the source we introduce a small parameter into the Navier-Stokes equation and construct corresponding asymptotic expansions. We consider the most popular model of real gases --- the van der Waals model, and ideal gases as well.
\end{abstract}

\maketitle

\section{Introduction}
Stationary three-dimensional flows of viscous fluid or gas are described by the following system (see \cite{Batch} for example):
\begin{equation}
\label{BasicEqs}
\left\{
\begin{aligned}
&\rho(\mathbf{u},\nabla)\mathbf{u}=-\nabla p+\eta\Delta\mathbf{u}+\left(\frac{\eta}{3}+\zeta\right)\nabla(\mathrm{div}\mathbf{u}),\\
&\mathrm{div}(\rho\mathbf{u})=0,
\end{aligned}
\right.
\end{equation}
where $\mathbf{u}(x)$ is the velocity field, $p(x)$ is the pressure, $\rho(x)$ is the density, $x\in\mathbb{R}^{3}$, $\eta$ and $\zeta$ are viscosity coefficients, which are assumed to be constants. The first equation in (\ref{BasicEqs}) is the Navier-Stokes equation corresponding to the momentum conservation law, the second one is the continuity equation responsible for the conservation of mass. In addition to (\ref{BasicEqs}) we take the isentropicity condition:
\begin{equation}
\label{Adiab}
(\mathbf{u},\nabla\sigma)=0,
\end{equation}
where $\sigma(x)$ is the specific entropy. This condition means that the specific entropy $\sigma(x)$ is constant along the stream lines.

Exact solutions of the Navier-Stokes equations corresponding to the jets issued from a source have been of interest since the middle of the 20th century. Landau seems to be one of the first who found a solution, which can be considered as a jet in an unrestricted incompressible medium, so-called \textit{submerged jet} \cite{Lan}. It is worth to say that incompressibility condition was crucial in his approach and moreover, thermodynamics of the flow was not considered. His solution has two significant drawbacks. The first one is vanishing of the flow for inviscid fluids and the second one is zero mass flux. Landau's solution was improved by Broman and Rudenko \cite{BrRud}. Using the incompressibility condition, they showed that Landau jet is the particular case of more general flows with nonzero mass flux and admitting the passage to inviscid fluids. Thermodynamics of such flows was investigated by Squire in \cite{Squ}. He considered the heat balance equation in addition to the Navier-Stokes system and found the distribution of the temperature in the jet. He supposed that the fluid was incompressible, but since the temperature changes along the flow, one has to take into account the changing of the density, which can be done by means of equation of state. Some of invariant solutions of the Navier-Stokes equations for incompressible viscous fluid and corresponding flows are studied in \cite{MRmlsd} by means of symmetries of differential equations \cite{KLR,VinKr}.

We can see that system (\ref{BasicEqs})-(\ref{Adiab}) is incomplete and in this paper we suggest using the equations of state of gases instead of incompressibility condition to make it complete. For this reason, we recall the geometrical description of thermodynamic states in a brief manner (see \cite{Lych,LRNon,DLT} for details). The first results in investigating one-dimensional Navier-Stokes flows with equations of state taken into account are obtained in \cite{GLRT}. Here, we restrict ourselves on a case of stationary isentropic flows in three-dimensional space depending on the distance from the source. In this case we are able to construct exact solutions for Euler equations. The solutions obtained are multivalued and different branches of these solutions are defined by different conditions at infinity. In case of viscous gases, we observe step-like solutions of the Navier-Stokes equations, and a smooth jump from one branch of Euler equation onto another occurs.
\section{Thermodynamic states}
In this section we recall some basic ideas in geometrical representation of thermodynamic states following \cite{Lych,LRNon,DLT}.

Let $(\mathbb{R}^{5},\theta)$ be a contact space equipped with coordinates $(e,v,T,p,\sigma)$, where $v=\rho^{-1}$ is the specific volume, $e$ and $T$ are the specific energy and the temperature respectively, and structure 1-form
\begin{equation*}
\theta=d\sigma-T^{-1}de-T^{-1}pdv.
\end{equation*}
By a thermodynamic state we mean a 2-dimensional Legendrian manifold $L\subset(\mathbb{R}^{5},\theta)$, such that $\theta|_{L}=0$. This condition means that the first law of thermodynamics holds on $L$ and if the specific entropy $\sigma$ is given in the form $\sigma=\sigma(e,v)$, the Legendrian manifold $L$ is defined completely by relations (equations of state):
\begin{equation}\label{Leg}\sigma=\sigma(e,v),\quad T=\frac{1}{\sigma_{e}},\quad p=\frac{\sigma_{v}}{\sigma_{e}}.\end{equation}
But in physics equations of state establish the connection between measurable quantities and are provided by experiments, that is why we have to eliminate the specific entropy $\sigma$ by means of projection
\begin{equation*}
\pi\colon\mathbb{R}^{5}\to\mathbb{R}^{4},\quad \pi\colon(e,v,T,p,\sigma)\mapsto(e,v,T,p).
\end{equation*}
Then the restriction $\overline{L}=\pi(L)$ is an immersed 2-dimensional Lagrangian manifold in a symplectic space $(\mathbb{R}^{4},\Omega)$, i.e. $\Omega|_{\overline{L}}=0$, where
\begin{equation*}\Omega=-d\theta=-T^{-2}dT\wedge de+T^{-1}dp\wedge dv-pT^{-2}dT\wedge dv,\end{equation*}
The condition $\Omega|_{\overline{L}}$ can be expressed as follows:
\begin{equation}\label{Comp}[f,g]=0\text{ on }\overline{L},\end{equation}
where functions $f(e,v,p,T)$ and $g(e,v,p,T)$ define the Lagrangian manifold $\overline{L}$:
\begin{equation}\label{Lagr}\overline{L}=\left\{f=0, g=0\right\},\end{equation}
and $[f,g]$ is the Poisson bracket with respect to the structure form $\Omega$:
\begin{equation*}[f,g]\hspace{1mm}\Omega\wedge\Omega=df\wedge dg\wedge\Omega.\end{equation*}
Summarising above discussion, thermodynamic state is either Legendrian manifold $L\subset(\mathbb{R}^{5},\theta)$ defined by (\ref{Leg}), or Lagrangian manifold $\overline{L}\subset(\mathbb{R}^{4},\Omega)$ defined by (\ref{Lagr}) with condition (\ref{Comp}). The last can be used to find a \textit{caloric} equation $g=e-\beta(v,T)$ if we know a \textit{thermic} one $f=p-\alpha(v,T)$ and can also be considered as a compatibility condition for the system
\begin{equation*}
f\left(e,v,\frac{1}{\sigma_{e}},\frac{\sigma_{v}}{\sigma_{e}}\right)=0,\quad g\left(e,v,\frac{1}{\sigma_{e}},\frac{\sigma_{v}}{\sigma_{e}}\right)=0,
\end{equation*}
defining the specific entropy $\sigma(e,v)$ up to a constant.

One may show, that not all the points on $L$ correspond to thermodynamic states. Namely, the differential quadratic form $\kappa$ is defined on $L$  \cite{Lych}:
\begin{equation*}
\kappa=d(T^{-1})\cdot de+d(pT^{-1})\cdot dv,
\end{equation*}
and domains on $L$, where $\kappa$ is negative, are applicable. The jumps between applicable domains preserving intensives (temperature, pressure and chemical potential) are phase transitions of the first type \cite{Lych,LRNon}. For further purposes we need the following theorem \cite{LRNon}:
\begin{theorem}
The Legendrian manifold $L$ is defined by the Massieu-Plank potential $\phi(v,T)$:
\begin{equation}\label{LegMas}\sigma=R(\phi+T\phi_{T}),\quad p=RT\phi_{v},\quad e=RT^{2}\phi_{T}.\end{equation}
The domain of applicable states is given by inequalities:
\begin{equation}\label{applic}\phi_{vv}<0,\quad \phi_{TT}+2T^{-1}\phi_{T}>0,\end{equation}
where $R$ is the universal gas constant.
\end{theorem}

\section{Formulation of the problem}
Suppose that we have an isotropic source at the origin and a thermodynamic state of the gas is given by the Massieu-Plank potential $\phi(v,T)$. Then, the condition $(\mathbf{u},\nabla\sigma)=0$ leads to the constancy of the specific entropy $\sigma=R\sigma_{0}$. This allows to express all the thermodynamic values in terms of $v$. Indeed, due to (\ref{LegMas}) we have:
\begin{equation}\label{eqforT}\sigma_{0}=\phi+T\phi_{T}.\end{equation}
Since (\ref{applic}) is valid, relation (\ref{eqforT}) defines implicitly $T(v)$. Using equations of state we can find $p(v)$ as well.

All the functions in (\ref{BasicEqs}) depend on the distance from the source $r=\sqrt{x_{1}^{2}+x_{2}^{2}+x_{3}^{2}}$, because the source is isotropic:
\begin{equation*}\mathbf{u}=U(r)\mathbf{r},\quad v=v(r).\end{equation*}
Assume that the intensity of the source $J$ is given. This means that the mass flux across the sphere of radius $r$ is equal to $J$:
\begin{equation}\label{flux}J=\int\limits_{S_{r}}v^{-1}(\mathbf{u},\mathbf{n})dS=4\pi r^{3}v^{-1}(r)U(r),\end{equation}
where $\mathbf{n}=\mathbf{r}/r$. From what follows a relation between $U(r)$ and $v(r)$:
\begin{equation}\label{Uv}U(r)=\frac{Iv(r)}{r^{3}},\end{equation}
where $I=J/4\pi$, and the continuity equation in (\ref{BasicEqs}) is satisfied.

We will consider flows in an unrestricted medium $\mathbb{R}^{3}$. For this reason we have to satisfy the requirement of regularity at infinity, i.e.
\begin{equation}\label{cond}\left.\mathbf{u}\right|_{r\to\infty}=0,\quad \left.v\right|_{r\to\infty}=v_{0}.\end{equation}

Thus, the problem can be formulated as follows. We are looking for the solution of (\ref{BasicEqs}) satisfying conditions at infinity (\ref{cond}) with additional requirement (\ref{flux}).
\section{Integration of Euler equations}
Let us first consider inviscid media, i.e. put $\eta=\zeta=0$ in (\ref{BasicEqs}). We get the following system:
\begin{equation}
\label{BasicEqs1}
\left\{
\begin{aligned}
&v^{-1}(\mathbf{u},\nabla)\mathbf{u}=-\nabla p,\\
&\mathrm{div}(v^{-1}\mathbf{u})=0,\\
&(\mathbf{u},\nabla\sigma)=0.
\end{aligned}
\right.
\end{equation}

\begin{theorem}
General solution of problem (\ref{BasicEqs1}), (\ref{cond}), (\ref{flux}) is given implicitly by the following formula:
\begin{equation}\label{solut}\frac{v^{2}}{2r^{4}}+I^{-2}f(v)-C_{0}=0,\end{equation}
where $C_{0}$ is a constant defined by means of (\ref{cond}), and
\begin{equation*}f(v)=\int vp^{\prime}(v)dv.\end{equation*}
\end{theorem}
\begin{proof}
Since we are looking for solutions depending on the distance from the source $r$, $\nabla=\mathbf{n}\partial_{r}$, and taking into account (\ref{Uv}), we get the following equation for $v(r)$:
\begin{equation*}-\frac{2v}{r^{5}}+\frac{v^{\prime}}{r^{4}}+\frac{p^{\prime}}{I^{2}}=0.\end{equation*}
Substituting $p=p(v(r))$, we get
\begin{equation*}\frac{d}{dr}\left(\frac{v^{2}}{2r^{4}}\right)+I^{-2}vp^{\prime}(v)\frac{dv}{dr}=0,\end{equation*}
from what follows (\ref{solut}).
\end{proof}

The above theorem provides a method of finding solutions of Euler equations for any equation of state.

\section{Examples}
In this section, we illustrate how above methods are applied to different models of gases. First of all, we elaborate ideal gases case and then discuss real ones.
\subsection{Ideal gases}
The Legendrian manifold $L$ for ideal gases is given by:
\begin{equation*}p=\frac{RT}{v},\quad e=\frac{nRT}{2},\quad \sigma=R\ln\left(T^{n/2}v\right),\end{equation*}
where $n$ is the degree of freedom.

The given level of the specific entropy $\sigma=R\sigma_{0}$ allows to express the pressure $p$ and the temperature $T$ as functions of the specific volume $v$:
\begin{equation*}p(v)=Rcv^{-(1+2/n)},\quad T(v)=cv^{-2/n},\end{equation*}
where $c=\exp(2\sigma_{0}/n)$.

Then (\ref{solut}) takes the following form:
\begin{equation*}\frac{v^{2}}{2r^{4}}+\frac{Rc(n+2)}{2I^{2}}v^{-2/n}-C_{0}=0,\end{equation*}
or in terms of density $\rho=v^{-1}$
\begin{equation}\label{solideal}\frac{1}{2r^{4}\rho^{2}}+\frac{Rc(n+2)}{2I^{2}}\rho^{2/n}-C_{0}=0.\end{equation}
\begin{theorem}
Solution $\rho(r)$ defined by (\ref{solideal}) exists if
\begin{equation*}r>\left(2\rho_{*}^{2}\left(C_{0}-RI^{-2}c(n/2+1)\rho_{*}^{2/n}\right)\right)^{-1/4},\end{equation*}
where
\begin{equation*}\rho_{*}=\left(\frac{2I^{2}nC_{0}}{Rc(n+1)(n+2)}\right)^{n/2}.\end{equation*}
\end{theorem}
Let us explore the asymptotic behavior of $\rho(r)$ if the density at infinity is given $\rho|_{r\to\infty}=\rho_{0}$.
\begin{theorem}
If $\rho_{0}=0$, then the asymptotic behavior of $\rho(r)$ at infinity is of the form:
\begin{equation*}
\rho(r)=\frac{1}{\sqrt{2C_{0}}r^{2}}+o\left(\frac{1}{r^{2}}\right),
\end{equation*}
if $\rho_{0}\ne0$, then
\begin{equation*}\rho(r)=\left(\frac{2C_{0}I^{2}}{Rc(n+2)}\right)^{n/2}+\sum\limits_{i=1}^{\infty}\frac{\beta_{i}}{r^{4i}}.\end{equation*}
\end{theorem}
Thus, constant $C_{0}$ is responsible for the density at infinity and we can see that there are two types of solutions defined by the condition at infinity. The graph for the density is represented in figure \ref{picDens}.

\begin{figure}[h!]
\centering
\includegraphics[scale=.25]{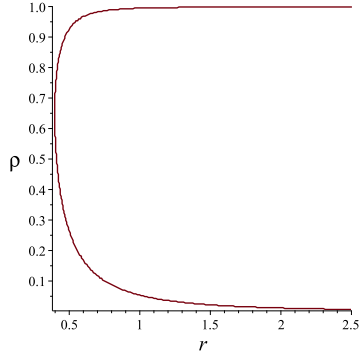}
\caption{Graph of function $\rho(r)$ for ideal gases}
\label{picDens}       % Give a unique label
\end{figure}

\subsection{van der Waals gases}
The most popular model in description of real gases is the van der Waals model. In this case the Legendrian manifold $L$ is given by state equations
\begin{equation*}
p=\frac{8T}{3v-1}-\frac{3}{v^{2}},\quad e=\frac{4nT}{3}-\frac{3}{v},\quad \sigma=\ln\left(T^{4n/3}(3v-1)^{8/3}\right).
\end{equation*}
Here we use reduced coordinates $(e,v,T,p,\sigma)$, i.e. the point $(1,1,1,1,1)$ is the critical point.

The pressure $p$ and the temperature $T$ are expressed in terms of the specific volume $v$ in the following way:
\begin{equation*}T(v)=c(3v-1)^{-2/n},\quad p(v)=\frac{8c}{(3v-1)^{1+2/n}}-\frac{3}{v^{2}},\end{equation*}
where $c=\exp(3\sigma_{0}/4n)$.

The solution is of the form:
\begin{equation}\label{solvdW}\frac{1}{2r^{4}\rho^{2}}+I^{-2}\left(4c(3\rho^{-1}-1)^{-(1+2/n)}\left(\rho^{-1}(n+2)-\frac{n}{3}\right)-6\rho\right)-C_{0}=0.\end{equation}
Using (\ref{solvdW}) it is easy to check that the following theorem is valid:
\begin{theorem}
If $\rho_{0}=0$, then the asymptotic behavior of $\rho(r)$ at infinity is of the form:
\begin{equation*}
\rho(r)=\frac{1}{\sqrt{2C_{0}}r^{2}}+o\left(\frac{1}{r^{2}}\right),
\end{equation*}
if $\rho_{0}\ne0$, then
\begin{equation*}\rho(r)=\rho_{0}+\sum\limits_{i=1}^{\infty}\frac{\beta_{i}}{r^{4i}},\end{equation*}
where $\rho_{0}$ and $C_{0}$ are connected by the relation
\begin{equation*}\frac{3C_{0}\rho_{0}I^{2}}{2}=2c(3\rho_{0}^{-1}-1)^{-(1+2/n)}(3n+6-n\rho_{0})-9\rho_{0}^{2}.\end{equation*}
\end{theorem}
The distribution of density for van der Waals gases has the same form as for ideal gases and is shown in figure \ref{picDensvdW}.

\begin{figure}[h!]
\centering
\includegraphics[scale=.25]{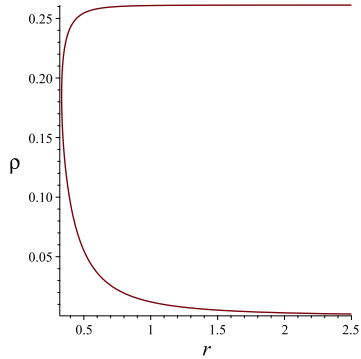}
\caption{Graph of function $\rho(r)$ for van der Waals gases}
\label{picDensvdW}       % Give a unique label
\end{figure}
The domain where solution exists can be derived by means of (\ref{solvdW}) in the same way as it was done for ideal gases.
\subsubsection{Phase transitions}
If the temperature of van der Waals gases is under the critical value, i.e. $T<1$, phase transitions may occur along the gas flow. Having a solution given by (\ref{solvdW}) it is possible to find the domains in space corresponding to different phases. We construct such distributions of phases for different levels of the specific entropy. From now and on, $s_{0}=\exp(\sigma_{0})$.

Suppose that the level is ``small'', $s_{0}=0.5$. The solution is multivalued and we start with its lower branch.

\begin{figure}[ht!]
\centering \subfigure[]{
\includegraphics[width=0.35\linewidth]{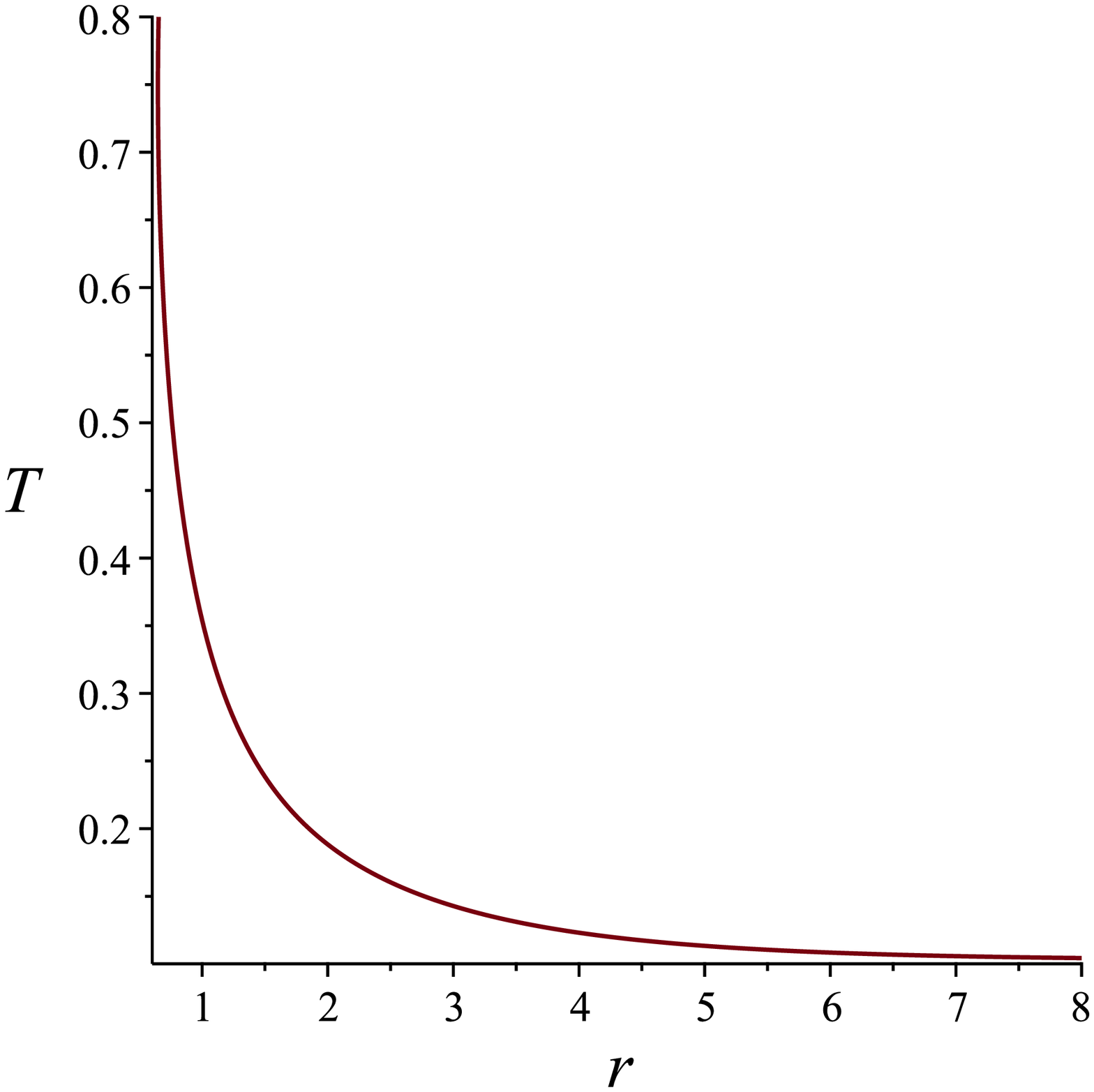} \label{templow1} }
\hspace{4ex}
\subfigure[]{ \includegraphics[width=0.35\linewidth]{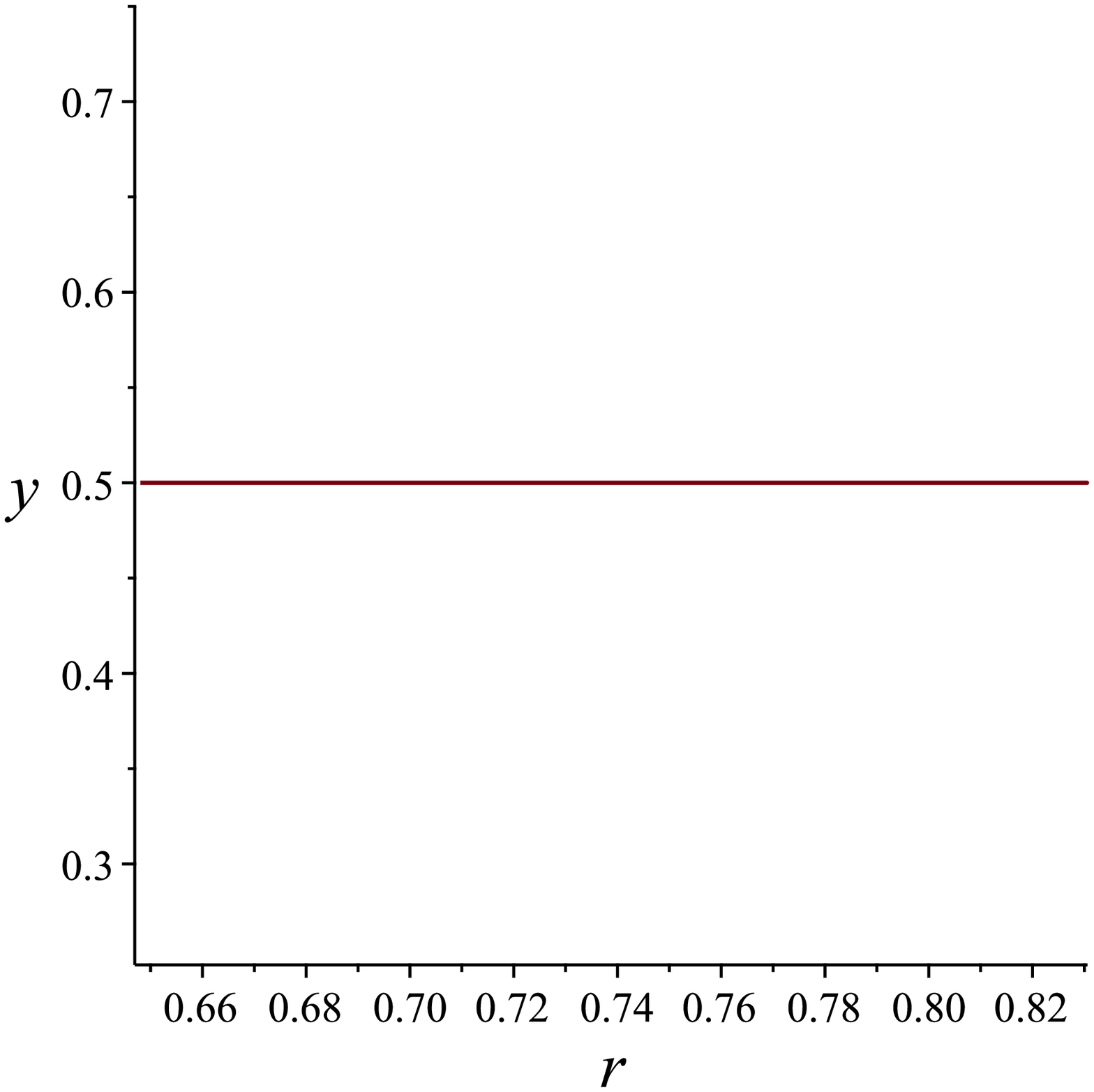} \label{phaselow1} }
\caption{\footnotesize{Dependence of temperature on the distance from the source and the distribution of phases in space for van der Waals gases: \subref{templow1} represents lower branch of solution; \subref{phaselow1} represents the corresponding phases. Variable $y=0.5$ means that the medium is in an intermediate state (condensation process).}} \label{lowercase1}
\end{figure}

Figure \ref{lowercase1} shows that if the temperature decreases along the gas flow, the gas is in an unstable state corresponding to condensation process.

\begin{figure}[ht!]
\centering \subfigure[]{
\includegraphics[width=0.35\linewidth]{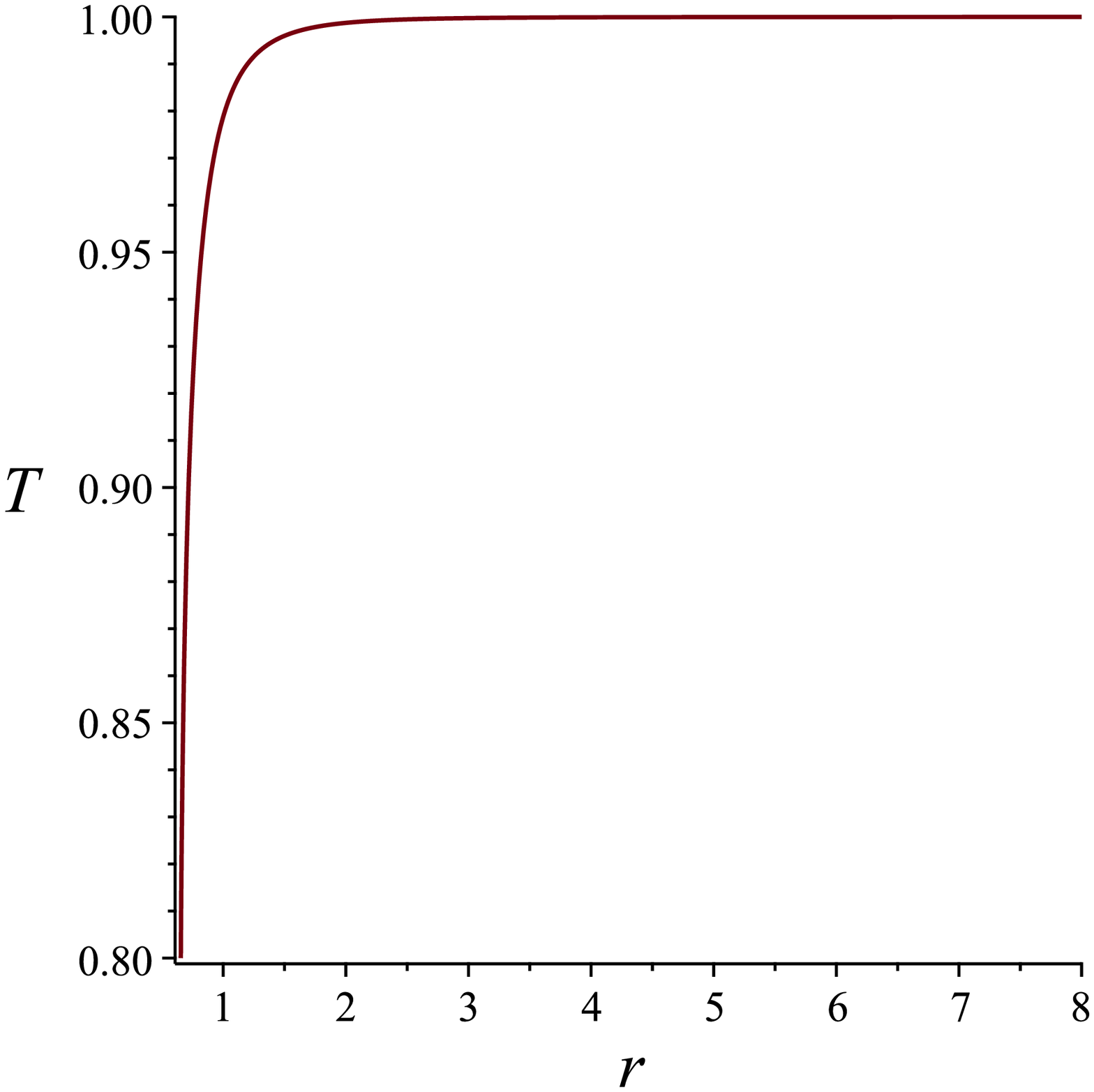} \label{temphigh1} }
\hspace{4ex}
\subfigure[]{ \includegraphics[width=0.35\linewidth]{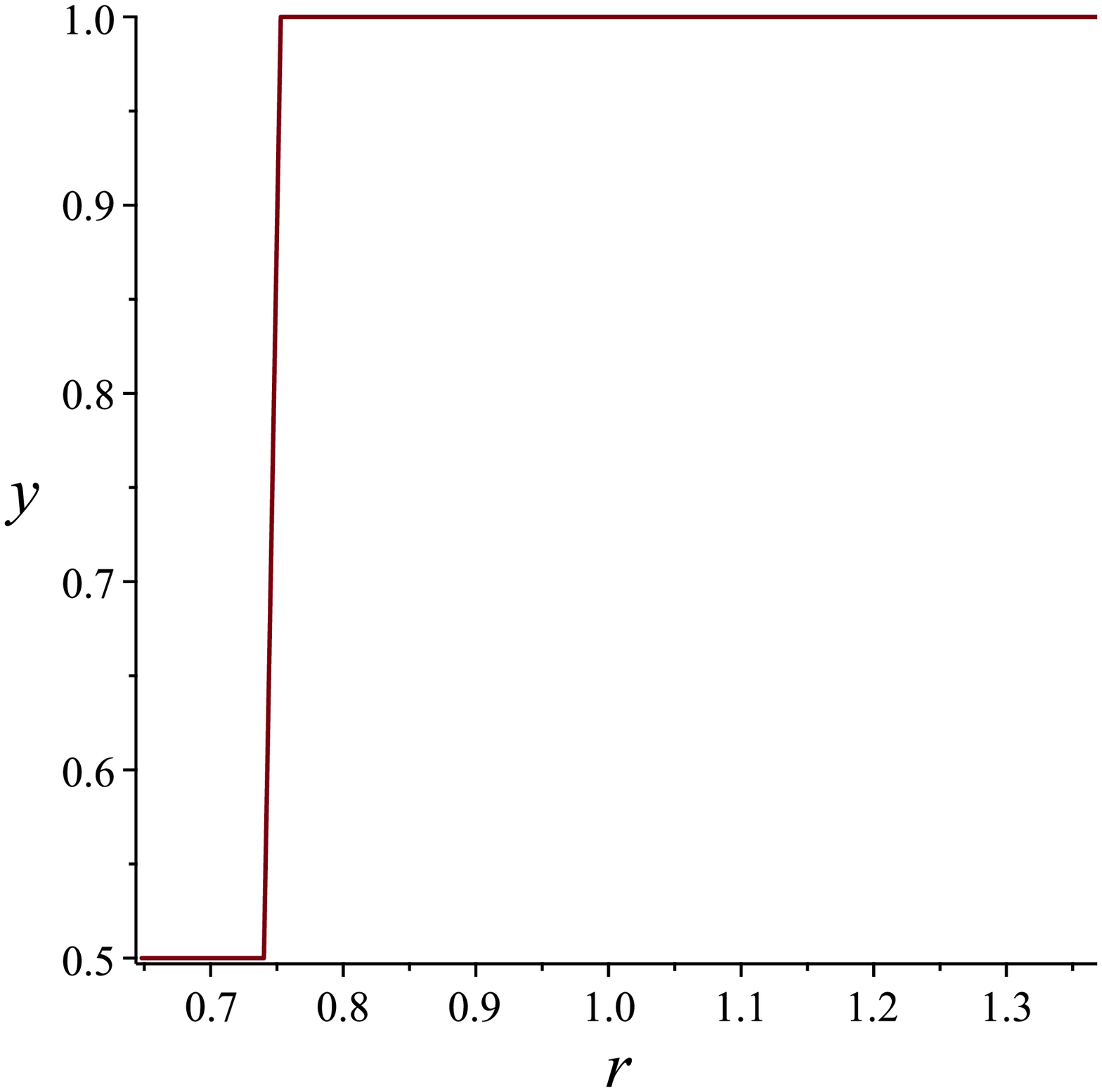} \label{phasehigh1} }
\caption{\footnotesize{Dependence of temperature on the distance from the source and the distribution of phases in space for van der Waals gases: \subref{temphigh1} represents higher branch of solution; \subref{phasehigh1} represents the corresponding phases. Variable $y=0.5$ means that the medium is in an intermediate state (condensation process), $y=1$ stands for the liquid phase}} \label{highercase1}
\end{figure}

Figure \ref{highercase1} shows that if the temperature increases while the distance increases, the condensation process is concentrated near the source and at large distance from the source the medium is in a liquid phase.

If the temperature $T$ at infinity is greater then critical value, i.e. $T_{0}>1$, the distribution of phases may be as in figure \ref{ph}.

\begin{figure}[h!]
\centering
\includegraphics[scale=.35]{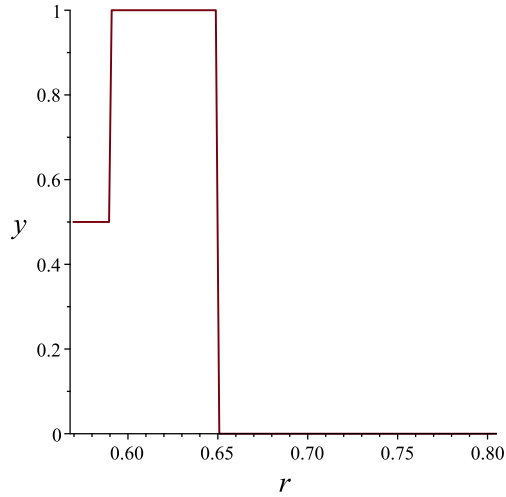}
\caption{\footnotesize{The distribution of phases in case $T_{0}=1.1$. If $y=1$, the medium is in a liquid phase, if $y=0.5$, the medium is in an intermediate state (condensation process), if $y=0$, the medium is in a gas phase}}
\label{ph}       % Give a unique label
\end{figure}

%%%%%%%%%%%%%%%%%%%%%%%%%%%%%
Suppose now that the level is ``big'', $s_{0}=200$.

\begin{figure}[ht!]
\centering \subfigure[]{
\includegraphics[width=0.35\linewidth]{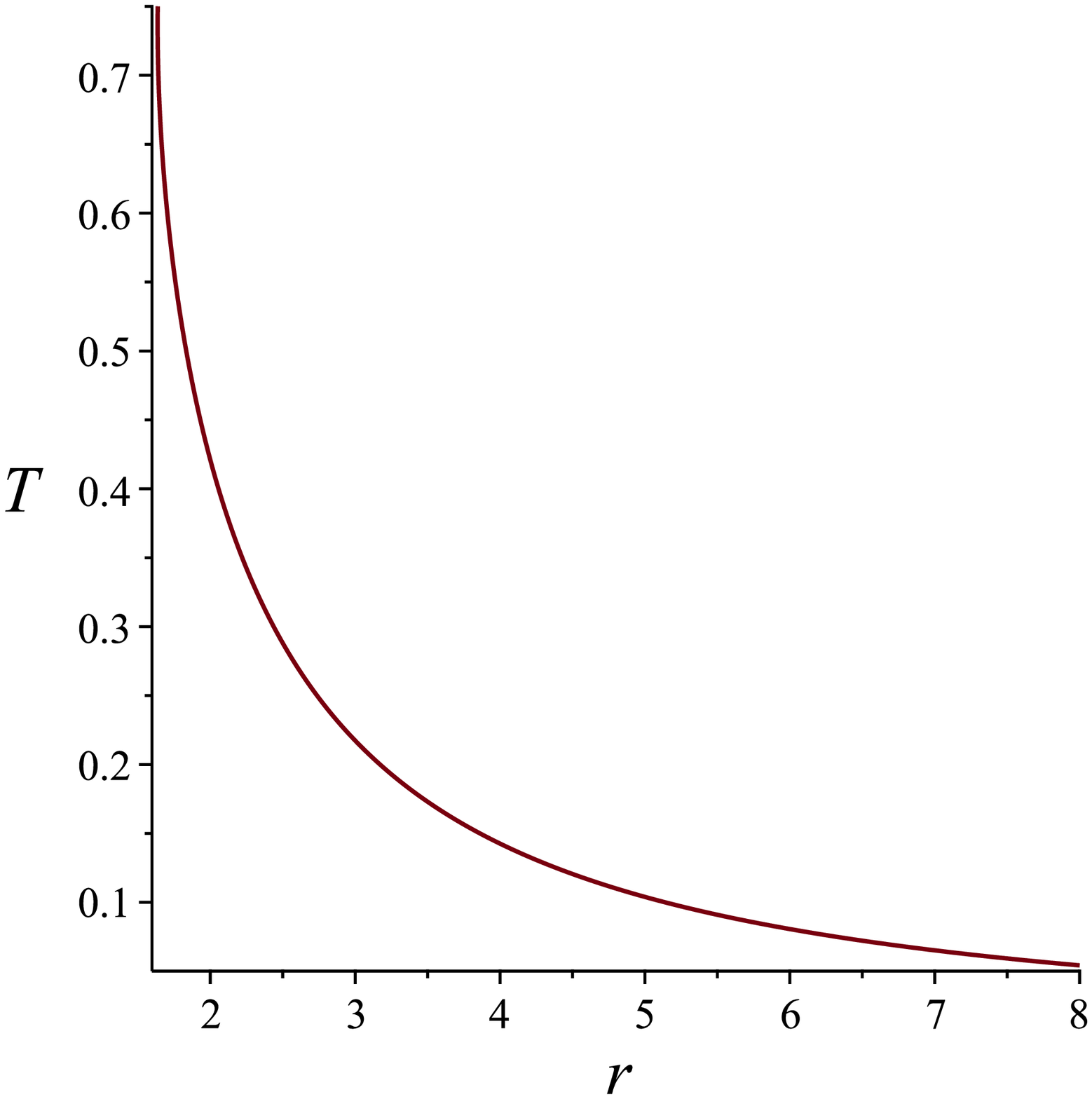} \label{templow2} }
\hspace{4ex}
\subfigure[]{ \includegraphics[width=0.35\linewidth]{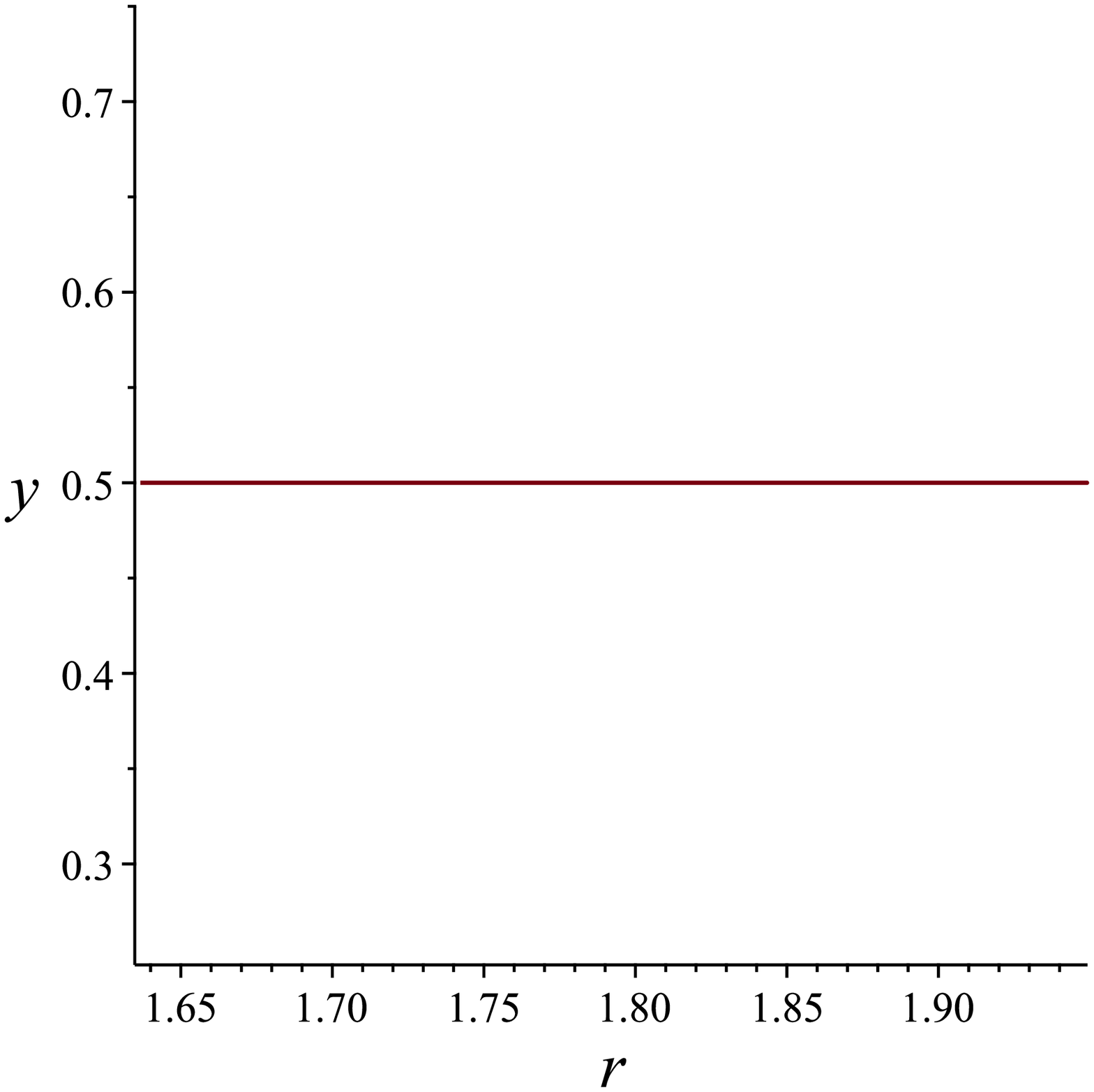} \label{phaselow2} }
\caption{\footnotesize{Dependence of temperature on the distance from the source and the distribution of phases in space for van der Waals gases: \subref{templow2} represents lower branch of solution; \subref{phaselow2} represents the corresponding phases. Variable $y=0.5$ means that the medium is in an intermediate state (condensation process)}} \label{lowercase2}
\end{figure}

Figure \ref{lowercase2} shows that if the temperature decreases along the gas flow, the gas is in an unstable state corresponding to condensation process and the picture is exactly the same as in the previous case.

\begin{figure}[ht!]
\centering \subfigure[]{
\includegraphics[width=0.35\linewidth]{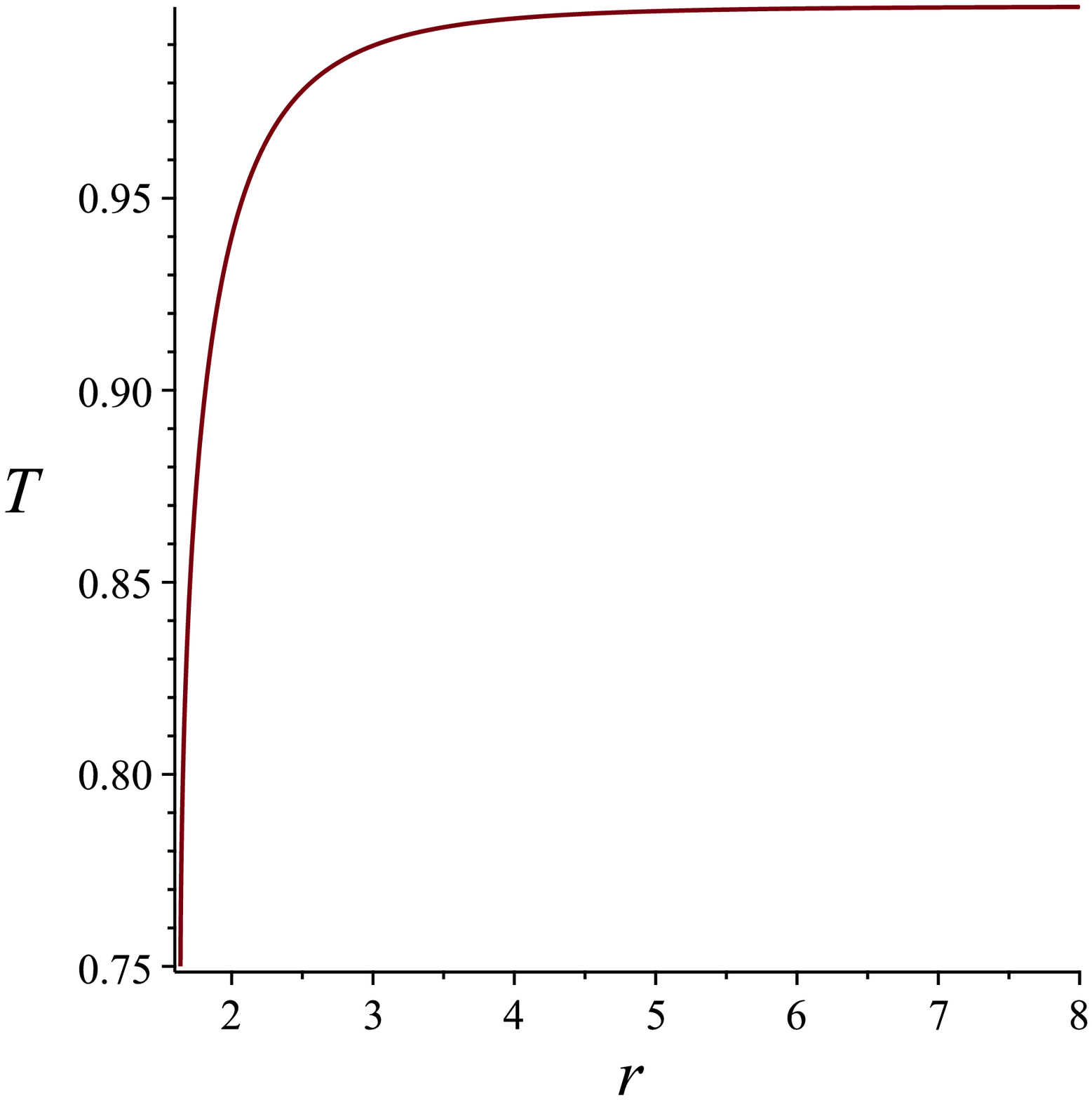} \label{temphigh2} }
\hspace{4ex}
\subfigure[]{ \includegraphics[width=0.35\linewidth]{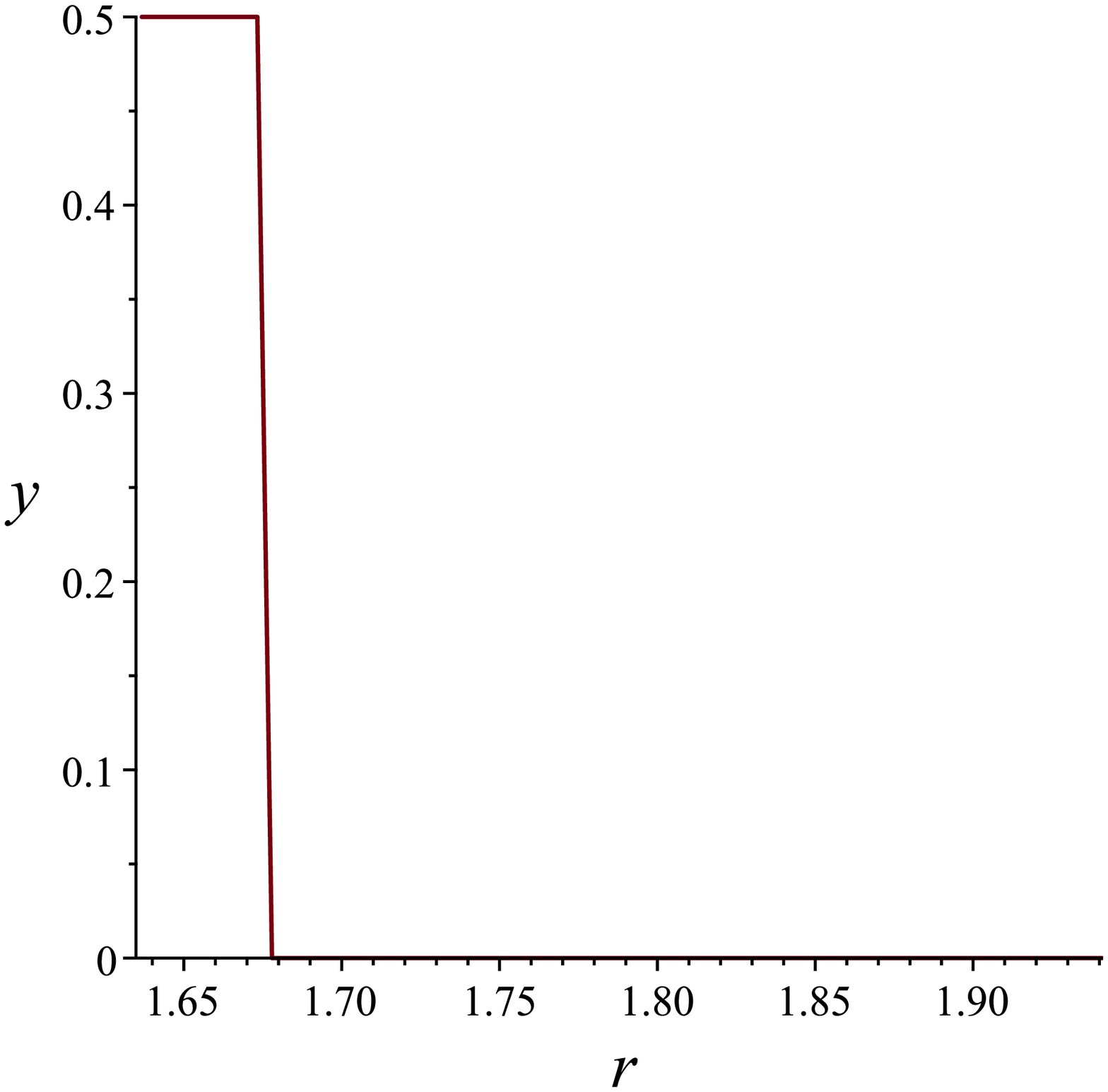} \label{phasehigh2} }
\caption{\footnotesize{Dependence of temperature on the distance from the source and the distribution of phases in space for van der Waals gases: \subref{temphigh2} represents higher branch of solution; \subref{phasehigh2} represents the corresponding phases. Variable $y=0.5$ means that the medium is in an intermediate state (condensation process), $y=0$ stands for the gas phase}} \label{highercase2}
\end{figure}

Figure \ref{highercase2} shows that the condensation process is concentrated near the source and at large distance from the source the medium is in a gas phase. The picture is very similar to filtration processes \cite{LRNon}.

%%%%%%%%%%%%%%%%%%%%%%%%

\section{Solutions of the Navier-Stokes equation}
\subsection{Singularly perturbed problem}
Let us now take into account the viscosity of the medium. In this case the equation for the specific volume $v(r)$ contains viscosity coefficients $\eta$ and $\zeta$ and can be written in the following form:
\begin{equation}\label{NS}-\frac{v}{r^{3}}(rv^{\prime\prime}-2v^{\prime})\mu+\frac{d}{dr}\left(\frac{v^{2}}{2r^{4}}+I^{-2}f(v)\right)=0,\end{equation}
where
\begin{equation*}\mu=I^{-1}\left(\zeta+\frac{4\eta}{3}\right)\end{equation*}
is considered to be a small parameter.

Note that in (\ref{NS}) the small parameter $\mu$ is multiplied by the higher derivative. Such problems are called \textit{singularly perturbed}. Their main feature is that their solutions do not converge to the solution of an unperturbed problem obtained by putting $\mu=0$. In \cite{VasButNef} it is shown that step-like solutions of such problem are typical, and methods of estimation of the location of the step are developed as well as the structure of asymptotic series for solution is suggested.

Numerical computations show that solution of (\ref{NS}) has a step-like form, and the bigger the viscosity $\mu$, the smoother the step. It is shown in figure \ref{NSpicsolut}.

\begin{figure}[h!]
\centering
\includegraphics[scale=.25]{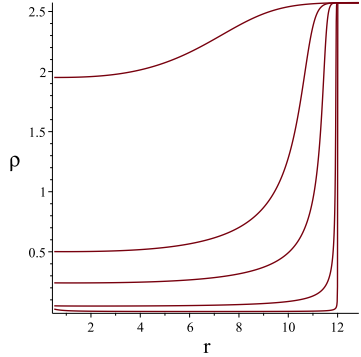}
\caption{The distribution of the density $\rho$ for ideal gas}
\label{NSpicsolut}       % Give a unique label
\end{figure}

\subsection{Asymptotics for the Navier-Stokes equation}
Here, we consider the intensity of the source $I$ as either small or big parameter. First of all, let us rewrite equation (\ref{NS}) in the form:
\begin{equation*}\frac{v}{r^{3}}(rv^{\prime\prime}-2v^{\prime})=\frac{I}{k}\frac{d}{dr}\left(\frac{v^{2}}{2r^{4}}+I^{-2}f(v)\right),\end{equation*}
where $k=4\eta/3+\zeta$. Introduce new variables $(x,w)$ in the following way:
\begin{equation*}r=I^{\alpha}x,\quad v=I^{\beta}w,\end{equation*}
where $\alpha$ and $\beta$ are constants. We get
\begin{equation}\label{modNS}\frac{w}{x^{3}}(xw^{\prime\prime}-2w^{\prime})=k^{-1}\frac{d}{dx}\left(\frac{I^{1-\alpha}w^{2}}{2x^{4}}+I^{3\alpha-2\beta-1}f\left(I^{\beta}w\right)\right).\end{equation}
Require that function $f(v)$ satisfies the inequality $|f(v)|\le Cv^{A}$. For both ideal and van der Waals gas this condition holds putting $A=-2/n$. In this case a small parameter can be introduced into the right hand side of (\ref{modNS}) by choosing $(\alpha,\beta)$ in the following way:
\begin{itemize}
\item if $I\ll 1$, then
\begin{equation*}
\left\{
\begin{aligned}
&3\alpha-2\beta-1+A\beta>0,\\
&1-\alpha>0.
\end{aligned}
\right.
\end{equation*}
\item if $I\gg 1$, then
\begin{equation*}
\left\{
\begin{aligned}
&3\alpha-2\beta-1+A\beta<0,\\
&1-\alpha<0.
\end{aligned}
\right.
\end{equation*}
\end{itemize}

\subsubsection{Ideal gases}
For ideal gases
\begin{equation*}f(v)=Rc\left(\frac{n}{2}+1\right)v^{-2/n}.\end{equation*}
Taking $\beta=n(n+1)^{-1}(2\alpha-1)$ and $\alpha=1/2$ for $I\ll 1$, $\alpha=2$ for $I\gg1$ we get:
\begin{equation}\label{NSperturb}\frac{w}{x^{3}}(xw^{\prime\prime}-2w^{\prime})=\varepsilon_{1,2}\frac{d}{dx}\left(\frac{w^{2}}{2kx^{4}}+f(w)\right),\end{equation}
where $\varepsilon_{1}=\sqrt{I}$ for $I\ll 1$ and $\varepsilon_{2}=1/I$ for $I\gg 1$.

We are looking for the solution of (\ref{NSperturb}) in the following form:
\begin{equation*}w(x)=w_{0}(x)+\varepsilon_{1,2}w_{1}(x)+\ldots,\end{equation*}
and for the zeroth $w_{0}(x)$ and the first $w_{1}(x)$ order terms we have expressions:
\begin{equation*}\begin{split}&w_{0}(x)=C_{1}x^{3}+C_{2},{}\\&
w_{1}(x)=-\frac{Rcn}{6C_{1}k}(C_{1}x^{3}+C_{2})^{-2/n}+\frac{1}{6xk}(2x^{4}C_{3}k-3C_{1}x^{3}+6kxC_{4}-3C_{2}),\end{split}\end{equation*}
where $C_{j}$ $(j=\overline{1,4})$ are constants.

\subsubsection{van der Waals gases}
For van der Waals gases we have:
\begin{equation*}f(v)=\frac{8c}{3(3v-1)^{2/n+1}}+\frac{4c(n+2)}{3(3v-1)^{2/n}}-\frac{6}{v}.\end{equation*}
We restrict our consideration on the case of monatomic gases, i.e. put $n=3$. If $I\gg1$ one may take $\alpha=2$, $\beta=6$ and get the following result:
\begin{equation*}w(x)=w_{0}(x)+\varepsilon w_{1}(x)+\ldots,\end{equation*}
where $\varepsilon=1/I$,
\begin{equation*}w_{0}(x)=C_{1}x^{3}+C_{2},\quad w_{1}(x)=\frac{1}{6xk}(2x^{4}C_{3}k-3C_{1}x^{3}+6C_{4}kx-3C_{2}),\end{equation*}
and $C_{j}$ $(j=\overline{1,4})$ are constants.

The case of $I\ll1$ can be elaborated in the same way.
\subsection{Regular asymptotic expansions}
Here, we consider $I$ as a small parameter and look for the solution of equation (\ref{NS}) written as
\begin{equation*}\frac{Iv}{r^{3}}(rv^{\prime\prime}-2v^{\prime})=k^{-1}\frac{d}{dr}\left(\frac{I^{2}v^{2}}{2r^{4}}+f(v)\right)\end{equation*}
in the following form:
\begin{equation*}v(r)=v_{0}(r)+Iv_{1}(r)+I^{2}v_{2}(r)+\ldots\end{equation*}
Then, the zeroth order term can be found from equation $f(v_{0})=f_{0}$, from what follows that
\begin{equation}\label{zero}v_{0}=f^{-1}(f_{0}),\end{equation}
where $f_{0}$ is a constant.

Invertibility conditions for $f(v)$ in case of ideal and van der Waals gases will be elaborated below.

The equation for the first order term is
\begin{equation*}(f^{\prime}(v_{0})v_{1}^{\prime}+f^{\prime\prime}(v_{0})v_{0}^{\prime}v_{1})r^{5}-kv_{0}r(r^{2}v_{0}^{\prime})^{\prime}=0.\end{equation*}
Since $v_{0}(r)=v_{0}=\mathrm{const}$, we have constant solution for the first order term as well:
\begin{equation}\label{first}v_{1}(r)=v_{1}.\end{equation}
One may show that solutions for the second $v_{2}(r)$ and the third order terms $v_{3}(r)$ are
\begin{equation}\begin{split}\label{twothree}&v_{2}(r)=\frac{v_{0}^{2}}{2f^{\prime}(v_{0})r^{4}}+\alpha_{1},{}\\&v_{3}(r)=\frac{2v_{0}}{r^{4}(f^{\prime}(v_{0}))^{2}}\left(\frac{kv_{0}^{2}}{r^{3}}-\frac{v_{1}}{4}\left(v_{0}f^{\prime\prime}(v_{0})-2f^{\prime}(v_{0})\right)\right)+\alpha_{2},\end{split}\end{equation}
where $\alpha_{i}$ are constants.

Thus, for any thermodynamic state we have found an asymptotic expansion for the Navier-Stokes equation and due to (\ref{zero})-(\ref{twothree}) the corresponding solution for the density $\rho=v^{-1}$ and for the velocity field is regular at infinity. If $I=0$, the velocity field vanishes due to (\ref{Uv}), while the density is constant.
\subsubsection{Ideal gases}
As we have seen, the invertibility problem of $f(v)$ needs to be investigated. In case of ideal gases the following theorem is valid:
\begin{theorem}
The function $f(v)$ for ideal gases is invertible and the corresponding zeroth order term is uniquely determined:
\begin{equation*}v_{0}=\left(\frac{2f_{0}}{Rc(n+2)}\right)^{-n/2},\end{equation*}
where $c=\exp(2\sigma_{0}/n)$.
\end{theorem}
\subsubsection{van der Waals gases}
For van der Waals gases we have the following result \cite{LRNon}:
\begin{theorem}The function $f(v)$ is invertible if the specific entropy constant $c=\exp(3\sigma_{0}/4n)$ satisfies the inequality:
\begin{equation}\label{inv}c>\frac{1}{4\alpha}(1+\alpha)^{1+\alpha}(2-\alpha)^{2-\alpha},\end{equation}
where $\alpha=1+2/n$.
\end{theorem}
Thus, if condition (\ref{inv}) holds, the zeroth order term is uniquely determined, otherwise, equation (\ref{zero}) has a number of solutions.

\par\bigskip\noindent
{\bf Acknowledgment.} This work was supported by the Russian Foundation for Basic Research (project 18-29-10013).

\bibliographystyle{amsplain}

%\noindent $\clubsuit$ Note to author:
%Proceedings articles should be formatted as in reference 1
%above, journal articles as in reference 2 above, and books as in
%reference 3 above.

\end{document}